%% file: main.tex
\documentclass[letterpaper, 10 pt, conference]{ieeeconf}  

\IEEEoverridecommandlockouts                              

\let\labelindent\relax
\usepackage{graphicx}      
\usepackage{booktabs}
\usepackage{cite}
\input{preamble.tex}
\usepackage{makecell}
 
\usepackage{lipsum}
\usepackage{comment}

\usepackage{hyperref}
\usepackage[dvipsnames]{xcolor}
\hypersetup{
	colorlinks=false,
	linkcolor=blue,
	filecolor=magenta,      
	urlcolor=cyan,
	pdftitle={POGD Regret and ILC},
	allbordercolors = OliveGreen,
	hidelinks = true
}
\makeatletter
\renewcommand*{\eqref}[1]{%
	\hyperref[{#1}]{\textup{\tagform@{\ref*{#1}}}}%
}
\makeatother

\title{ \LARGE \bf
    Regret Analysis of Online Gradient Descent-based \\
    Iterative Learning Control with Model Mismatch
}

 \author{Efe C. Balta, Andrea Iannelli, Roy S. Smith, John Lygeros
 	\thanks{
 	Research supported by the Swiss National Science Foundation under: NCCR Automation, a National Centre of
 		Competence in Research (grant number 180545), and grant number 200021\_178890. 
 		}
 	\thanks{All authors are with the Automatic Control Laboratory, ETH Zurich, 8092, Zurich, Switzerland.
 		{\tt\footnotesize \{ebalta,iannelli,rsmith,lygeros\}@control.ee.ethz.ch }}
 }

\usepackage{algpseudocode}

\begin{document}

\maketitle
\begin{abstract}
	In Iterative Learning Control (ILC), a sequence of feedforward control actions is generated at each iteration 
	on the basis of partial model knowledge and past measurements with the goal of steering the system toward a desired reference trajectory.
	This is framed here as an online learning task, where the decision-maker takes sequential decisions by solving a sequence of optimization problems having only partial knowledge of the cost functions. Having established this connection, the performance of an online gradient-descent based scheme using inexact gradient information is analyzed in the setting of dynamic and static regret, standard measures in online learning. Fundamental limitations of the scheme and its integration with adaptation mechanisms are further investigated, followed by numerical simulations on a benchmark ILC problem.
\end{abstract}
\section{Introduction}

Online learning-based optimization approaches have been increasingly studied in recent literature\cite{agarwal2021regret,dixit2019online,goel2020regret,muthirayan2021online}. The online-learning setting usually assumes an unknown cost function that changes at each time-step, and an optimization algorithm that aims to minimize the unknown cost by using any prior information, e.g., a model, and observations of the cost and/or the gradient at each time-step. 
A natural generalization of this online learning setting is to consider an online-learning control problem, where the decision maker aims to control a dynamical systems while minimizing a control cost at each time step. One of the first works recognizing the connection between online-learning and adaptive control was \cite{Raginsky_OCO_Adapt}.
Since then, there have been many works focusing on solving the online-learning control problem under various assumptions on the type of model, uncertainty, constraints, and noise characteristics \cite{Hazan_nonStochastic,muthirayan2021online,goel2020regret}. 
Regret is 
a common metric in many of the online-learning problems, as it provides a characterization of cost incurred at each time step due to unknown changes to the cost function or problem structure. Additionally, since a fixed point convergence is not well-defined in many cases of online-learning problems, regret provides an alternative metric to assess the effectiveness of a given algorithm.

The class of online convex optimization (OCO) methods have been widely used for online learning problems~\cite{hazan2016introduction}. Among the family of OCO methods, online gradient descent is of specific interest due to its simplicity and favorable guarantees on achievable regret under mild assumptions on the cost function and constraints~\cite{hazan2007logarithmic}.
However, many online gradient descent algorithms assume access to gradient observations, which may not be available in many practical control applications.
Recent work has considered variants of the online gradient descent using inexact gradient information for proximal-type optimization algorithms in an online setting~\cite{dixit2019online} with additive errors on the gradient. 
Iterative approaches for control in an inexact gradient setting are studied in~\cite{agarwal2021regret}, where only additive errors to the known dynamics are considered.

Online optimization problems have a close relationship with Iterative Learning Control (ILC) methods. 
In ILC, the controller utilizes an input-output model of the process and learns from past iterations dealing with iteration-invariant~\cite{barton2010norm,amann1996iterative,mishra2010optimization,liao2022robustness} as well as iteration-varying problems~\cite{altin2017iterative,balta2020switch,yu2017robust,van2016optimality}. 
While convergence properties under various assumptions on the dynamics and model uncertainty have been analyzed, the regret analysis in an online learning ILC setting has not been considered yet in the literature.
This work proposes an online-learning based ILC method which utilizes a preconditioned online gradient descent method in the presence of model mismatch.
After formulating the proposed control algorithm, its static and dynamic regret are quantified and variants are discussed and investigated. Our general analysis encompasses common ILC schemes previously proposed in the literature, and thus their regret characterization is an additional outcome of the work.
The contribution of this work is therefore threefold: (i) a new online learning-based ILC methodology inspired by online gradient descent methods, (ii) a detailed regret analysis of the proposed ILC method and its variants, and (iii) regret analysis of existing ILC methods from the literature as special instances of the proposed ILC method.

Section~\ref{sec:probForm} formulates the problem and proposes the online ILC controller.
Section~\ref{sec:Reg} provides a detailed analysis of regret in the transient and limit cases, while
Section~\ref{sec:ilc} 
extends the results to the iteration-invariant ILC methods from the existing literature. Section~\ref{sec:demo} provides a numerical demonstration and Section~\ref{sec:conc} gives concluding remarks.

\textit{Notation}: Given a square matrix $A$, $\|A\|_2$ denotes its spectral norm and $||A||_P = ||P^{1/2}AP^{-1/2}||_2$, where $P$ is a symmetric positive definite matrix of appropriate dimension. Given a vector $x$, the weighted norm is $||x||_P = \sqrt{x^TPx}$.

\section{Problem Formulation}\label{sec:probForm}

The considered iterative learning control problem is modelled by the following input-output dynamics in the absence of exogenous disturbances
\begin{equation}
    y(x_k) = H_k x_k,
\end{equation}
where, $y_k\in \reals^{n_y}$ is the output, and $x_k\in \reals^{n_x}$ is the input at iteration $k$.
The input-output dynamics map $H_k$ is commonly employed in the ILC literature and is referred to as the \emph{lifted} representation of a system. 
Concretely, $H_k$ may represent the temporal evolution of a periodic linear parameter or time varying, or invariant dynamics along an iteration, which may require specific uncertainty structures~\cite{barton2010norm,amann1996iterative,mishra2010optimization,liao2022robustness}. Alternatively, spatial models, reduced from their spatiotemporal partial differential equations forms may be utilized, as commonly done in spatial ILC applications~\cite{hoelzle2015spatial}.

Crucially, $H_k$ is only partially known, with an uncertainty structure formally stated below, and a nominal estimate $\textstyle{M\approx H_k}$ is available.
In each iteration $k$, the goal is to minimize
\begin{equation}
\label{eq:ilc_cost}
    f_k(x) =  \frac{1}{2}\left( || H_k x - r||^2_Q +||x||_R^2\right),
\end{equation}
where $r$ is a reference to be tracked, $\textstyle{Q = Q^T \succ 0}$ is a weighting matrix, and the second term with $\textstyle{R = R^T \succ 0}$ is used for regularization. This term is a flexible design choice used to penalize undesired features of the solution, such as high inputs.
The weighting matrices $Q,R$ may also be positive semi-definite in certain cases, see~\cite{liao2022robustness}.
Note that following the same formulation, iteration varying and a priori known $r_k$ may be used in place of $r$. We focus on the case with iteration invariant $r$ in this work for simplicity.
The gradient of \eqref{eq:ilc_cost} is given by 
\begin{equation}
\label{eq:true_grad}
    \nabla f_k(x) = H^T_kQ(H_k x - r) + Rx.
\end{equation}
Notice that while the term $H_k x$ can be evaluated directly by running an iteration on the true system with the input $x$ and measuring the output $y(x)$, the adjoint dynamics of the true system $H_k^T$ are unknown. 
To circumvent this problem, one can use the nominal estimate, $M$, to estimate the gradient, leading to
\begin{equation}
\label{eq:approx_grad}
	\tilde{\nabla} f_k(x) =M^TQ (y(x)-r) + Rx.
\end{equation}
The ILC update applied to generate new inputs at each iteration is given by the following Preconditioned Online Gradient Descent (POGD) step
\begin{equation}
\label{eq:ilc_update}
    x_{k+1} = \Pi_{\U}^{W}\left( x_k - \alpha_k W^{-1}\tilde \nabla f_k(x_k) \right),
\end{equation}
where $\textstyle{W\!=\!W^T\!\succ\! 0}$ is a preconditioner matrix, $\U$ is a convex input constraint set, and $\alpha_k$ is the step-size at iteration $k$.
The concrete uncertainty representation of $H_k$ is stated next.
\begin{assumption}
	\label{assm:model}
	For all $k$, the true dynamics $H_k$ belongs to the set $\bm{H}(M,\Delta) :=\{H | \; H = M + M\Delta  \}$, where $M$ is a nominal estimate with full column rank and the uncertainty $\Delta$ belongs to the unstructured norm bounded set $\bm{\Delta}(W,\gamma) :=\{\Delta | \; || \Delta ||_W \leq \gamma  \}$, where $\gamma\geq 0$ is the known uncertainty size, and $W$ is the preconditioner matrix.
\end{assumption}
Uncertainty representations similar to Assumption~\ref{assm:model} have been used in the past literature \cite{barton2010norm,van2009iterative,son2015robust}.
The projection operator to the set $\U$ in the weighted preconditioner norm is 
\begin{equation}
    \Pi_{\U}^{W}(x) := \argmin_{u\in \U} ||u-x||_{W}.
\end{equation}
The main technical contribution of the paper is the analysis of this POGD-ILC in terms of two notions of regret. The most general case corresponds to the dynamic regret \cite{Jadbabaie_ICAIS_15}
\begin{align}
	J_d(T) = \sum_{k=1}^T f_k(x_k) - \sum_{k=1}^T f_k(x^*_k),
\end{align}
where $\textstyle{x^*_k =   \argmin_{x\in\U}f_k(x)}$, i.e., the regret with respect to an iteration-wise optimal control policy.
Additionally, we consider the traditional static regret~\cite{hazan2007logarithmic}
\begin{align}\label{def:stat_regret}
	J_s(T) = \sum_{k=1}^T f_k(x_k) - \min_{x\in\U}\sum_{k=1}^T f_k(x).
\end{align}
The static regret is with respect to a controller that defines a single fixed optimal input with the hindsight information about the full sequence of iteration-varying $f_k$. 
The regret analysis is based on the following assumptions:
\begin{assumption}\label{assm:Lips}
	For each $k$, $f_k$ is locally Lipschitz continuous in $\U$ with Lipschitz constant $L_k$ in the weighted preconditioner norm, i.e., $||f_k(x)-f_k(y)|| _W \leq L_k ||x - y||_W$ $\forall$ $x,y \in \U$; moreover, $\textstyle{\bar{L} := \sup_k\{L_k \} < \infty}$.
\end{assumption}
\begin{assumption}\label{assm:input}
	The optimal input between consecutive iterations is bounded as $||x_k^*-x_{k+1}^*||_W\leq e_k$.
\end{assumption}
\begin{assumption}
\label{Ass_sigma}
	There exist a sequence $\sigma_k$ such that $\textstyle{||W^{-1} (M\Delta)^TQ(H_k x^*_k - r) ||_W \leq \sigma_k}$ for all $k$ with $\textstyle{\bar{\sigma} := \sup_k\{\sigma_k \}<\infty}$.
\end{assumption}
Assumption~\ref{assm:Lips} holds for example when $\U$ is compact.
Assumption~\ref{assm:input} ensures the optimal inputs are bounded and an upper bound estimate is available. In practice, for strongly convex $f_k$, or compact $\U$, this assumption is already met, in which case $e_k$ can be taken as the normed difference itself.
Assumption~\ref{Ass_sigma} is due to the model mismatch term $\Delta$, and characterizes the distance between the fixed point of \eqref{eq:ilc_update} for fixed $k$, and the optimizer $x_k^*$. We formally show how this term appears in some of the regret bounds and discuss its role under various settings in later sections.
Finally, we define
\begin{equation*}
	\phi_k := ||I - \alpha_k W^{-1} (M^TQH_k+R)||_W,\quad \Phi_{j,k} := \prod_{i = j}^k \phi_{i}
\end{equation*}

\section{Regret Analysis}\label{sec:Reg}

In this section we analyze the dynamic and static regrets of the sequential actions taken using the POGD algorithm~\eqref{eq:ilc_update}.
In Section \ref{ss:dyn_T}, an upper bound on the dynamic regret valid under the standing assumptions is provided, followed by a study of the regret's limit behaviour when $T \rightarrow \infty$. 
Additional assumptions under which the regret is shown to be sublinear are then discussed in Section \ref{ss:dyn_T_adapt}, before analyzing the static regret case in Section \ref{ss:stat}. 

\subsection{Dynamic Regret: Transient and Asymptotic Behavior}\label{ss:dyn_T}
The following theorem provides an upper bound on the dynamic regret of the POGD algorithm under the design choices and assumptions discussed so far.

\begin{theorem}[Dynamic Regret of POGD-ILC]
	\label{thm:main_dynamic2}
	Under Assumptions \ref{assm:model}, \ref{assm:Lips}, \ref{assm:input}, and \ref{Ass_sigma}, consider the choice of preconditioner $\textstyle{W = M^TQM + R}$ and define $\textstyle{w := || W^{-1}M^TQM||_W}$. If $w\gamma <1$
	and the step-size is chosen as $\alpha_k \in \left(0,\frac{2}{1+w\gamma}\right)$,  
	then the dynamic regret of POGD is upper bounded by
	\begin{align*}\label{eq:dyn_reg_UB}
		J_d(T) \!\leq\!  \bar{L}\delta_{x_1}\! \sum_{k=1}^T\!\Phi_{1,k}\!+\!\bar{L} \bar{\sigma} \!\sum_{k=1}^T\sum_{j=1}^k \alpha_j \Phi_{j+1,k}\! +\! \bar{L}\sum_{k=1}^T E_k 
	\end{align*}
	where $E_k := \sum_{j=1}^k e_j \Phi_{j+1,k}$ and $\delta_{x_1}:=||x_1-x^*_1||_W$. 
\end{theorem}
\begin{proof}
	We first bound the distance between the input updates and the corresponding optimal inputs.
	\begin{equation}
		\begin{aligned}
			\label{eq:iter_Dyn}
			||x_{k+1} &- x^*_{k+1}||_W \leq ||x_{k+1} - x^*_k||_W + ||x^*_k - x^*_{k+1}||_W \\
			&=||\Pi_{\U}^{W}\left( x_k - \alpha_k W^{-1}\tilde \nabla f_k(x_k) \right) \\
			&\quad- \Pi_{\U}^{W}\left( x^*_k - \alpha_k W^{-1}\nabla f_k(x^*_k) \right) ||_W +e_k \\
			& \leq ||x_k - \alpha_k W^{-1} \left(M^TQ(H_k x_k - r) +Rx_k\right)\\
			&~+x^*_k - \alpha_k W^{-1} \left(H_k^TQ(H_k x^*_k - r) +Rx^*_k\right)||_W + e_k \\
			& \leq ||\left(I - \alpha_k W^{-1}\left( M^TQH_k+R\right)\right)(x_k-x^*_k) ||_W \\
			&\quad + \alpha_k ||W^{-1} (M\Delta)^TQ(H_k x^*_k - r) ||_W + e_k\\
			& \leq \phi_k ||x_k-x^*_k||_W + \alpha_k\sigma_k + e_k,
		\end{aligned}
	\end{equation}
	where in the first inequality we use the triangle inequality 
	and in the second equality the fact that $x^*_k$ is a fixed point of the POGD with the true gradient $\nabla f_k$, $x^*_k=\Pi_{\U}^{W}\left( x^*_k - \alpha_k W^{-1}\nabla f_k(x^*_k) \right)$, and Assumption \ref{assm:input}.
	For the other inequalities we use the fact that the weighted projection operator is nonexpansive in the weighted preconditioner norm, Cauchy-Schwartz inequality and Assumption \ref{Ass_sigma}. 
	Next, we show the step-size parameters required to ensure $\phi_k<1$. Using Assumption \ref{assm:model} and the specific choice of preconditioner we have 
	\begin{align}
		\label{eq:bound_phi_k}
		\phi_k &= ||I - \alpha_k W^{-1} (M^TQ(M + M\Delta)+R)||_W   \notag \\
		& = ||(1-\alpha_k)I - \alpha_k W^{-1}M^TQM\Delta) ||_W \notag \\
		& \leq |1 -\alpha_k| + \alpha_k || W^{-1}M^TQM\Delta||_W \notag \\
		& \leq |1 -\alpha_k| + \alpha_kw\gamma,
	\end{align}
	where $||I||_W = 1$ was used in the first inequality.
	To ensure that $\phi_k < 1$, $\alpha_k$ must be chosen such that
	\begin{equation}\label{eq:bound_phi_k_small_1}
		|1 - \alpha_k| < 1-  \alpha_kw\gamma.
	\end{equation}
	Then since $w\gamma <1$, $\alpha_k\in \left(0,\frac{2}{1+w\gamma}\right)$ implies $\phi_k <1$.
	By iterating \eqref{eq:iter_Dyn} one gets
	\begin{align*}
		||x_{k+1} - x^*_{k+1} ||_W &\leq ||x_1-x^*_1||_W \prod_{j=1}^k \phi_j\\
		&\quad + \sum_{j=1}^k \left(\left(\sigma_j\alpha_j + e_j\right) \prod_{i = j+1}^k \phi_{i}\right) ,
	\end{align*}
	where we adopt the convention $\prod_{j+1}^j a_j= 1$. 
	Using the Lipschitz constant $L_k$ we get 
	\begin{align*}
		f_{k+1}(&x_{k+1})- f_{k+1}(x^*_{k+1})\leq L_k ||x_{k+1} - x^*_{k+1} ||_W \\
		&\leq L_k \bigg( ||x_1-x^*_1||_W\prod_{j=1}^k \phi_j +  \bar{\sigma}\sum_{j=1}^k \alpha_j \prod_{i = j+1}^k \phi_{i} \\
		&\qquad +\sum_{j=1}^ke_j \prod_{i = j+1}^k \phi_{i},\bigg).
	\end{align*}
	Taking the sum from $1$ to $T$ and using the upper bound $\bar{L}$ instead of $L_k$ at each step gives the desired result.
\end{proof}

The condition $w\gamma <1$ can be fulfilled by choice of the regularization matrix $R$. 
To see this, define $w_1$ and $\gamma_1$ the values of $w$ and $\gamma$ associated with $R_1$ and $W_1$.
If $w_1 \gamma_1 \geq 1$, we can always find $R_2$ such that $w_2 \gamma_2 < 1$.
This is because, from the choice of preconditioner $W$ and the definition of $w$ in Theorem~\ref{thm:main_dynamic2}, $w$ scales approximately with $||W^{-1}||$ and thus is $\mathcal{O}(||R^{-1}||)$. On the other hand, a valid (possibly not the tightest) upper bound on the uncertainty size $\gamma$ is $\mathcal{O}(1)$ in $R$. Consider without loss of generality $\textstyle{R_1 = \rho_1I}$ and $\textstyle{R_2 = \rho_2I}$, with $\textstyle{\rho_2 > \rho_1>0}$.
Using the definitions, we have $\textstyle{||\Delta||_{W_1} \leq \cond(W_1^{1/2}) ||\Delta|| = \gamma_1}$, where $\cond(\cdot)$ denotes the condition number of the matrix.
Since $\textstyle{\rho_2 > \rho_1}$, $\textstyle{\cond(W_2) < \cond(W_1)}$ and thus $\textstyle{||\Delta||_{W_2} \leq \cond(W_2^{1/2}) ||\Delta||=\gamma_2 < \gamma_1}$, i.e. $\gamma_1$ is still a valid uncertainty size for the new choice of $R_2$.

Using the upper bound obtained in Theorem~\ref{thm:main_dynamic2}, we characterize the asymptotic behavior of the dynamic regret. 
\begin{corollary} [Average Regret of POGD-ILC]
	\label{cor:dyn_regret_asymp}
	Under the same conditions as Theorem~\ref{thm:main_dynamic2}, if $\alpha_k=\alpha_0 k^{-c}$, with $\alpha_0 \in \left(0,\frac{2}{1+w\gamma}\right)$, $0\leq c<1$, then 	
	\begin{align}
		\lim\limits_{T \to \infty}	\frac{J_d(T)}{T}  \leq \mathcal{O}(1)+\frac{\bar{L}\sum_{k=1}^T E_k}{T}.
	\end{align}
\vspace{0.01cm}
\end{corollary}
\begin{proof}
From Theorem~\ref{thm:main_dynamic2}, $J_d(T)$ is bounded by
\begin{align}
\underbrace{\bar{L} \delta_{x_1} \sum_{k=1}^T\Phi_{1,k}}_{\text{Term I}}+\underbrace{\bar{L} \bar{\sigma}\alpha_0 \sum_{k=1}^T\sum_{j=1}^k  j^{-c} \Phi_{j+1,k}}_{\text{Term II}} + \underbrace{\bar{L}\sum_{k=1}^T E_k}_{\text{Term III}}
\end{align}
where in Term II the explicit expression of the step size has been used. 
Term I can be interpreted as the contribution to the regret due to distance of the initial decision from the optimal one. Define $\textstyle{\bar{\phi}_k := \sup_i\{\phi_i\}_{i=1}^k}$ and recall that $\textstyle{\bar{\phi}_k <1}$ by \eqref{eq:bound_phi_k},~\eqref{eq:bound_phi_k_small_1}, and the choice of the step size. Then
\begin{equation*}
	 \bar{L} \delta_{x_1} \sum_{k=1}^T\Phi_{j,k}\leq \bar{L} \delta_{x_1} \sum_{k=1}^T \bar{\phi}_k^k\leq \bar{L} \delta_{x_1} \sum_{k=1}^T \bar{\phi}_T^k\leq  \frac{\bar{L}\delta_{x_1}}{1-\bar{\phi}_T},
\end{equation*}
where we used the monotonicity of $\bar{\phi}_k$ in the second inequality, and the upper bound of the infinite sum in the last inequality.
From (\ref{eq:bound_phi_k}) there exists a finite $\bar{T}$, which depends on $\alpha_0$ and $c$, such that for $T>\bar{T}$
\begin{equation}
	 \bar{\phi}_T \leq 1 + (w\gamma-1) \alpha_0 T^{-c} 
\end{equation}
and thus
\begin{equation}
\lim\limits_{T \to \infty} \frac{\bar{L}\delta_{x_1}}{T(1-\bar{\phi}_T)}\leq \lim\limits_{T \to \infty} \frac{\bar{L}\delta_{x_1} T^{c}}{T(w\gamma-1) \alpha_0}=0
\end{equation}
whenever $c<1$.
Next, consider Term II and define $S_k:=\sum_{j=1}^k  j^{-c} \Phi_{j+1}^k$, which thus describes the growth of this term at each step $k$. Observe that
\begin{equation}\label{eq:s_dyn}
S_{k+1}=\phi_{k+1} S_k+\frac{1}{(k+1)^c}
\end{equation}
where we know from \eqref{eq:bound_phi_k} that $\phi_{k+1}<1$. For our choice of $\alpha_k$, two cases should be considered. When $0<c<1$ (i.e., vanishing step size), $\phi_{k+1} \rightarrow 1$ for $k \to \infty$. In the limit $k \to \infty$, the sequence $S_k$ will thus converge to a finite constant value $S_\infty$.
When $c=0$, $\phi_{k+1}<1$ as $k \to \infty$, and thus $S_k$ can be bounded between zero and the trajectory of an asymptotically stable linear system with constant input of $1$.
Therefore, by using the asymptotic behavior of the linear time varying system~\eqref{eq:s_dyn} we are able to characterize the asymptotic behavior of the regret for Term II.
As a result, in both cases Term II achieves linear regret
\begin{equation}
\lim\limits_{T \to \infty} \frac{\bar{L} \bar{\sigma}\alpha_0 \sum_{k=1}^T S_k}{T}\leq \mathcal{O}(1)
\end{equation}
\end{proof}
It is worth noting that the presented case generalizes some of the results from existing literature. As an example, \cite{dixit2019online} presents a similar result for the fixed step-size case and strongly convex cost functions which corresponds to the case with $c=0$ and appropriately chosen $R$.
Additionally, the interpretation of the regret bound in terms of the dynamical equation \eqref{eq:s_dyn} provides additional insights in terms of algorithm design and provides a basis for developing system-level synthesis-type regret optimal design~\cite{martin2022safe}. 

Corollary \ref{cor:dyn_regret_asymp} shows that the POGD algorithm applied to the ILC with model mismatch does not lead to a sublinear regret. The latter is regarded as a favorable property for sequential decision making algorithms, because it suggests that on the average the decisions asymptotically converge to the optimal ones at each stage. Convergence is prevented here by two terms, namely Term II and Term III. 
Term III is known as \emph{complexity} \cite{Hall_Willett_ICML_2013} or \emph{regularity} \cite{Jadbabaie_ICAIS_15} term in the dynamic regret literature and captures the effect of the temporal variability of the optimal sequence of actions. It is well-known that an upper-bound on the dynamic regret will have an explicit dependence on it and, in this setting, little can be said on its growth without prior information or assumptions on $H_k$.
By inspecting the derivation of the second term of the right hand-side in the bound \eqref{eq:iter_Dyn}, Term II is the contribution to the regret due to the suboptimality of the direction taken to update the decision at $k$. More precisely, this term is related to the term upper bounded by $\sigma_k$ in Assumption \ref{Ass_sigma} and is zero only if there is no model mismatch (i.e. $M=H_k$). 

\subsection{Adaptive POGD Algorithm}\label{ss:dyn_T_adapt}

Leveraging the previous observations and the proof of Corollary \ref{cor:dyn_regret_asymp}, modifications to the original POGD algorithm which are sufficient for achieving sublinear regret of Term II are discussed next. First, the required new assumptions are stated and discussed.
\begin{assumption}
\label{assm:model_2}
For all $k$, the true dynamics $H_k$ belongs to the set $\bm{H}_k(M_k,\Delta_k) :=\{H | \; H = M_k + M_k\Delta_k  \}$. $M_k$ is a full column rank nominal estimate at $k$ and the uncertainty $\Delta_k$ belongs to the unstructured norm bounded set $\bm{\Delta}_k(W,\gamma_k) :=\{\Delta | \; || \Delta ||_W \leq \gamma_k  \}$, where $\gamma_k\leq \gamma$ for all $k$ 
and $\gamma_k \to 0$ as $k \to \infty$.
\end{assumption}
This Assumption is a stronger version of Assumption \ref{assm:model} and requires the uncertainty size to asymptotically vanish. This could be achieved, for example, with an online identification scheme providing updated estimates of the model $M_k$ and of the uncertainty based on input-output measurements $\{(y_i,x_i)\}_{i=1}^T$ gathered during the decision making problem. Asymptotic convergence to zero of the estimation error $|| H_k - M_k||_W$ would also require appropriate excitation conditions on $r$ in the spirit of recursive parameter identification schemes used in adaptive control \cite{astrom2008adaptive}.

\begin{assumption}
\label{Ass_sigma_2}
	There exist $\tilde \sigma_k$ such that 
	$||W^{-1} (M_k\Delta_k)^TQ(H_k x^*_k - r) ||_W \leq \tilde\sigma_k$ for all $k$. 
	Moreover, $\tilde \sigma_k \to 0$ as $k \to \infty$.
\end{assumption}
This Assumption replaces Assumption \ref{Ass_sigma} and redefines the sequence of upper bounds $\tilde \sigma_k$ for the case when the estimate $M_k$ changes across iterations. The asymptotic behavior of $\tilde \sigma_k$ is a consequence of Assumptions \ref{assm:model_2}.
Further, define
\begin{equation*}
	\tilde \phi_k := ||I - \alpha_k W^{-1} (M_k^TQH_k+R_k)||_W,\quad \tilde \Phi_{j,k} := \prod_{i = j}^k \tilde \phi_{i}
\end{equation*}

Consider now an adaptive variation of the POGD algorithm described in Section \ref{sec:probForm} which, leveraging Assumption \ref{assm:model_2}, uses for its decisions the updated estimate of the model $M_k$.  The following Corollary shows that the associated dynamic regret is sublinear if the \emph{complexity} term is sublinear.
\begin{corollary}[Average Regret with Adaptation]
	\label{cor:dyn_regret_asymp_2}
	Under Assumptions \ref{assm:Lips}, \ref{assm:input}, \ref{assm:model_2}, and \ref{Ass_sigma_2}, consider the choice of preconditioner $W = M_1^TQM_1 + R_1$, with $R_1$ chosen so that 
		$w_k\gamma <1$ for all $k$, where $w_k := || W^{-1}M_k^TQM_k||_W$. Consider also the regularizer weighting matrix $\textstyle{R_k = W - M_k^TM_k}$. If the step-size is chosen as $\alpha_k=\alpha_0$ 
		with $\alpha_0 \in \left(0,\frac{2}{1+w\gamma}\right)$, then
		\begin{align}
\lim\limits_{T \to \infty}	\frac{J_d(T)}{T}  \leq\frac{\bar{L}\sum_{k=1}^T E_k}{T}.
	\end{align}
\vspace{0.01cm}
\end{corollary}
\begin{proof}
Following the derivations in \eqref{eq:iter_Dyn}, the distance between the input updates and the corresponding optimal inputs is
		\begin{align}
			&||x_{k+1} - x^*_{k+1}||_W \notag\\
			& \leq ||\left(I - \alpha_0 W^{-1}\left( M_k^TQH_k+R_k\right)\right)(x_k-x^*_k) ||_W\notag \\
			& ~+ \alpha_0 ||W^{-1} (M_k\Delta_k)^TQ(H_k x^*_k - r) ||_W + ||x^*_k - x^*_{k+1}||_W\notag\\
			& \leq \tilde \phi_k ||x_k-x^*_k||_W + \alpha_0\tilde \sigma_k + e_k,\label{eq:iter_Dyn_2}
		\end{align}
where Assumptions \ref{assm:input} and \ref{Ass_sigma_2}, the choice of $R_k$ and constant step-size were used. Using Assumption \ref{assm:model_2} and the choice of $R_k$ we further get
	\begin{align}
	\label{eq:bound_phi_k_tilde}
		\tilde \phi_k &= ||(1-\alpha_0)I - \alpha_0 W^{-1}M_k^TQM_k\Delta_k) ||_W   \notag \\
		& \leq |1 -\alpha_0| + \alpha_0 || W^{-1}M_k^TQM_k\Delta_k||_W \notag \\
		& \leq |1 -\alpha_0| + \alpha_0 w_k\gamma_k\leq |1 -\alpha_0| + \alpha_0 w_k\gamma,
	\end{align}
Then since $w_k\gamma <1$, $\alpha_0 \in \left(0,\frac{2}{1+w\gamma}\right)$, there exists $\hat \phi <1$ such that $\tilde \phi_k <\hat \phi$ for all $k$.
	Similarly to the proof of Theorem \ref{thm:main_dynamic2}, we then get
	\begin{align*}
		J_d(T)\! \leq \! \bar{L} \delta_{x_1} \sum_{k=1}^T\Phi_{j,k}\!+\!\bar{L} \alpha_0 \sum_{k=1}^T\sum_{j=1}^k \tilde \sigma_j \Phi_{j+1,k}\! +\! \bar{L}\sum_{k=1}^T E_k 
	\end{align*}
Note that Term I and Term III are unchanged, thus the former has again a sublinear growth because the constant step size satisfies the conditions of Corollary \ref{cor:dyn_regret_asymp}. As for Term II, the key difference is that now the stepsize is constant and $\sigma_k$ is kept inside the inner summation. As a result, the variable $\tilde S_k:=\sum_{j=1}^k  \tilde \sigma_k \tilde \Phi_{j+1,k}$ describing the growth of Term II at each step $k$ is such that
\begin{equation}\label{new_S_dyn}
\tilde S_{k+1}\leq \hat \phi \tilde S_k+\tilde \sigma_k 
\end{equation}
where, from Assumption \ref{Ass_sigma_2}, $\textstyle{\sigma_k \to 0}$ as $\textstyle{k \to \infty}$. By iterating \eqref{new_S_dyn} it is seen that the sequence $\tilde S_k$ converges to $\textstyle{\tilde S_\infty=0}$ and thus Term II achieves sublinear regret. 
\end{proof}
Compared to the originally considered POGD algorithm, the adaptive version features three major changes: the model estimate is updated on line; the step size is kept constant (non-diminishing); the regularization matrix is adapted as a function of the current model estimate. Note that at this stage this is not a complete algorithm, as it needs to be complemented by an online identification algorithm satisfying Assumption \ref{assm:model_2}. The purpose of its presentation is primarily to establish conditions on this complementary identification procedure to make the commonly used POGD algorithm competitive from a regret perspective. 

\subsection{Static Regret}\label{ss:stat}
Whereas dynamic regret provides a powerful metric for analyzing the performance of an online learning algorithm, its upper bound depends on the limiting behavior of the complexity term, Term III, which is unknown in general. This term disappears in the static regret case \eqref{def:stat_regret}, which is studied next. 
As a reminder, the fixed input $\textstyle{x^*}$ computed in hindsight to minimize the sum of observed costs, i.e. $\textstyle{x^* = \arg\min_{x\in\U}\sum_{k=1}^T f_k(x)}$, see~\eqref{def:stat_regret}.
The analysis is based on the following assumption.
\begin{assumption}
\label{Ass_eta}
	There exist $\eta_k$ such that $|| W^{-1/2} \tilde\nabla f_k(x^*)|| \leq \eta_k$ for all $k$, and $\textstyle{\bar{\eta} = \sup_k\{\eta_k\}}$.
\end{assumption}
\begin{corollary}
	\label{cor:static_regret}
Under the conditions of Theorem~\ref{thm:main_dynamic2} and Assumption \ref{Ass_eta}, the static regret of POGD is upper bounded by	
	\begin{align}\label{eq:stat_bound}
		J_s(T) \leq  \bar{L} \delta_{x_1} \sum_{k=1}^T\Phi_{j,k}+\bar{L} \bar{\eta} \sum_{k=1}^T\sum_{j=1} \alpha_j \Phi_{j+1,k}
	\end{align}
\end{corollary}
\begin{proof}
We start by bounding the distance between POGD solutions and the optimal solution $x^*$ for the static regret. 
	\begin{align}
		||x_{k+1} &- x^*||_W = ||\Pi_{\U}^{W}\left( x_k - \alpha_k W^{-1}\tilde \nabla f_k(x_k) \right) - x^*||_W\notag\\
					&\leq ||x_k - \alpha_k W^{-1}\tilde \nabla f_k(x_k) - x^*||_W \notag\\
		& \leq ||x_k - \alpha_k W^{-1} \left(M^TQ(H_k x_k - r) +Rx_k\right)\notag\\
		&~-x^* + \alpha_k W^{-1} \left(M^TQ(H_k x^* - r) +Rx^*\right)\notag\\
					&~-\alpha_k W^{-1} (M^TQ(H_k x^* - r)+Rx^*) ||_W \notag\\
		& \leq ||\left(I - \alpha_k W^{-1}\left( M^TQH_k+R\right)\right)(x_k-x^*_k) ||_W \notag \\
		&\quad + \alpha_k ||W^{-1} \left(M^TQ(H_k x^* - r)  + Rx^* \right) ||_W \notag \\
		& \leq \phi_k ||x_k-x^*||_W + \alpha_k||W^{-1} \tilde\nabla f_k(x^*)||_W\notag \\
		& \leq \phi_k ||x_k-x^*||_W + \alpha_k|| W^{-1/2} \tilde\nabla f_k(x^*)||,\label{eq:static_reg}
	\end{align}
where similar to the proof of Theorem~\ref{thm:main_dynamic2} we use non-expansive property of the weighted projection and triangle inequality to derive the desired result. In the last step, we use the definition of the weighted norm to represent the second term in the induced matrix 2-norm. 
Starting from \eqref{eq:static_reg}, we get the upper bound in \eqref{eq:stat_bound} by following the same steps as the proof of Theorem~\ref{thm:main_dynamic2}, where $\bar{\sigma}$ is here replaced by $\bar{\eta}$, since the $\phi_k$ terms are identical in both proofs.
\end{proof}
Following the arguments of Corollary~\ref{cor:dyn_regret_asymp}, it can be seen that the static regret grows linearly due to the new Term II (now depending on $\bar{\eta}$). In the interest of space, detailed discussions on the static regret and the modifications needed to achieve sublinearity are omitted here as they follow closely the dynamic regret counterparts.
Note that the optimality of $x^*$ is not necessary for the proof of Corollary~\ref{cor:static_regret}. 
Therefore, the static regret defines a worst-case cost optimality gap against any static policy played over the iteration horizon of $T$. 
This observation provides further insight on the meaning of regret, and its distinction with respect to other metrics such as optimality gap or convergence rates. 

\section{The Iteration Invariant Problem}\label{sec:ilc}
In this section, we specialize the results of Theorem~\ref{thm:main_dynamic2} to more commonly considered ILC settings featuring the assumption on constant cost function, i.e., $\textstyle{f_k(x) = f(x)}$ for all iterations $k$. 
Specifically, we assume that $\textstyle{f(x) = \frac{1}{2}\left( || H x - r||^2_Q +||x||_R^2\right)}$, 
where the true dynamics $ H $ has the same uncertainty description defined in Assumption \ref{assm:model} but is now time-invariant, i.e., $H = M + M\Delta$ for a fixed realization $\Delta$ in all iterations. This results in the ILC update
\begin{equation}
	\label{eq:ilc_normal}
	x_{k+1} = \Pi_{\U}^{W}\left( x_k - \alpha_k W^{-1}\tilde \nabla f(x_k) \right),
\end{equation}
which now has a fixed point $\bar{x}$ under suitable conditions, see~\cite{liao2022robustness,son2015robust}.
Following our analysis in the proof of Theorem~\ref{thm:main_dynamic2}, it is easy to see that the step-size rule given in the theorem with the given preconditioner choice ensures convergence to the fixed point, i.e., $\textstyle{||x_{k+1} -\bar{x}||_W \leq \phi_k ||x_k - \bar{x}||_W}$ (Proof omitted in the interest of space).
Additionally, note that the fixed point $\bar{x}$ is not necessarily the optimal point $x^*$ due to the model mismatch, thus $||\bar{x}-x^*||$ is nonzero in the general case (see~\cite{liao2022robustness} for a detailed characterization of the ILC fixed point).
The ILC iteration \eqref{eq:ilc_normal} is a specific instance of the main POGD algorithm given in \eqref{eq:ilc_update}. 
Therefore, the regret analysis in this section follows the results from previous sections.

The ILC update \eqref{eq:ilc_normal} with $\textstyle{\U = \reals^{n_x}}$ and $\textstyle{\alpha_k=1}$ results in a variant of the controller commonly referred as norm-optimal ILC under suitable preconditioner matrix design~\cite{barton2010norm,van2016optimality,van2009iterative,hoelzle2015spatial,balta2020switch}, while the case of convex $\U \subset \reals^{n_x}$ with a suitably chosen fixed step-size $\alpha_k = \bar{\alpha}$ is a variant of the optimization-based ILC~\cite{liao2022robustness,son2015robust}.
The output dynamics in the form $\textstyle{y(x_k) = Hx_k + b}$, with iteration invariant unknown offset $b$ is considered in~\cite{hoelzle2015spatial,barton2010norm,rezaeizadeh2017iterative}, for the case with $\textstyle{\U = \reals^{n_x}}$ and $\textstyle{\alpha_k=1}$. 
Furthermore, the case of iteration varying but bounded $b$ results in a bounded input bounded output stability condition~\cite{barton2010norm}, or more generally characterized as an input-to-state stability property in~\cite{liao2022robustness}.
The robust performance of similar unconstrained ILC algorithms under various uncertainty representations is discussed in~\cite{van2009iterative}, analyzing the robust monotonic convergence conditions via tools from the robust control literature. 
Due to the model mismatch, and also for the cases with bounded disturbance, such ILC algorithms achieve nonzero asymptotic error. Therefore, it is desirable to design control parameters to minimize the asymptotic ratio (gain) of the fixed point mismatch $||\bar{x}-x^*||$ to the uncertainty size in the problem, e.g., size of the uncertainty set or the disturbance set.
Here, we proceed with the output dynamics of the form $\textstyle{y(x_k) = Hx_k }$ for simplicity and draw conclusions about the interpretation of regret in the presence of the fixed point mismatch $||\bar{x}-x^*||$.

Since here we have $f_k(x) = f(x)$ for all iterations $k$, the constant input $\textstyle{x^* = \arg\min_{x\in\U}f(x)}$ is the optimal action for both the dynamic and static problems, thus the associated notions of regret coincide and will be referred to as \emph{ILC regret}. 
Therefore, we only investigate related corollaries of Theorem~\ref{thm:main_dynamic2} and provide additional results for the case with model learning, as in Corollary~\ref{cor:dyn_regret_asymp_2}. Additionally, we provide insights on how the optimality of the ILC fixed point has an interpretation for certain cases and how the implication of model learning on optimality and regret differ.
An auxiliary lemma is presented first, followed by the ILC regret results.
\begin{lemma}
	\label{lem:sum_bound}
	For any $\{a_i\}_{i=1}^T$ and $c\in(0,1)$, the following upper bound holds.
	\begin{align*}
		S_T = \sum_{k=1}^T\sum_{j=1}^kc^{k-j}a_j \leq \frac{1}{1-c}\sum_{j=1}^Ta_j.
	\end{align*}
\end{lemma}
\begin{proof}
	Expanding the outer sum and writing out the resulting partial sums results in
	\begin{align*}
		S_T &=  a_1 + \textstyle{\sum_{j=1}^2c^{2-j}a_j + \ldots + \sum_{j=1}^Tc^{T-j}a_j }\\
		&= a_1 + a_1c + a_2 +  \ldots + a_1c^{T-1} + a_2c^{T-2} + \ldots + a_T \\
		& = a_1 (1+c^2+\ldots+c^{T-1}) \\
		& \quad +a_2 (1+c^2+\ldots+c^{T-2}) + \ldots + a_T \\
		&\leq (1-c)^{-1}(a_1 + a_2 + \ldots + a_T),
	\end{align*}
	where we used the sum of the infinite series as an upper bound, since $c\in(0,1)$.
\end{proof}
\begin{proposition}[ILC regret for~\eqref{eq:ilc_normal}]\label{prop:ilc_normal}
	Under Assumptions \ref{assm:model} and \ref{assm:Lips}, consider the choice of preconditioner $W = M^TQM + R$ and define $w := || W^{-1}M^TQM||_W$. If $w\gamma <1$
	and a constant step-size is chosen as $\alpha=\alpha_0 \in \left(0,\frac{2}{1+w\gamma}\right)$, the static/dynamic regret for the controller update \eqref{eq:ilc_normal} is bounded by
	\begin{align}
		J_{\mbox{ILC}}(T) \leq  \frac{\bar{L} \left(\delta_{x_1} + \sigma \alpha_0 T\right) }{1-\phi},  
	\end{align}
	where $\phi = ||I - \alpha_0 W^{-1} (M^TQH+R)||_W$, $\delta_{x_1}:=||x_1-x^*||_W$, 
	and $\sigma \geq 0$ is such that $||W^{-1} (M\Delta)^TQ(H x^* - r) ||_W \leq \sigma$.
\end{proposition}
\begin{proof}
	The bound derived for $J_{\mbox{ILC}}(T)$ follows directly from the proof of Theorem~\ref{thm:main_dynamic2}, with constant $x^*$ instead of $x^*_k$. Recognizing the iteration invariant problem with identical cost functions $f(x)$, we can follow~\eqref{eq:iter_Dyn} to get 
	\begin{align}
		\label{eq:ilc_normal_a0}
		||x_{k+1} &- x^*||_W \leq \phi ||x_k-x^*||_W + \alpha\sigma.
	\end{align}
	Then, the resulting summation can be shown to be 
	\begin{align}
		J_{\mbox{ILC}}(T) \leq  \bar{L} \delta_{x_1} \sum_{k=1}^T\phi^k+\bar{L} \sigma \alpha_0 \sum_{k=1}^T\sum_{j=1}^k \phi^{k-j}.
	\end{align}
	Finally, by using the sum of the infinite series for the first term and Lemma~\ref{lem:sum_bound} for the second term, we obtain the desired result. 
\end{proof}
The bound for the iteration varying $\alpha_k$ follows similarly from Theorem~\ref{thm:main_dynamic2} and is omitted here. The following characterizes the asymptotic behavior of the limit.

\begin{corollary}\label{cor:ilc_normal_asymp}
	The ILC update \eqref{eq:ilc_normal} with time-varying and constant step-size selections 
	achieve constant average regret
	\begin{align}
		\lim\limits_{T \to \infty}	\frac{J_{\mbox{ILC}}(T)}{T}  \leq \mathcal{O}(1)
	\end{align}
\vspace{0.01cm}
\end{corollary}
The proof of Corollary~\ref{cor:ilc_normal_asymp} is omitted in the interest of space. It builds on the bounds given in Proposition~\ref{prop:ilc_normal} and uses the arguments adopted in Theorem~\ref{thm:main_dynamic2} and Corollary~\ref{cor:dyn_regret_asymp} specialized for a fixed cost function.

The linear regret is due to the mismatch term $||W^{-1} (M\Delta)^TQ(H x^* - r) ||_W$, which characterizes the distance $d = ||\bar{x} - x^*||$.
Therefore, we see a clear relationship between convergence and regret in this case. 
For example, if we have no model mismatch, i.e., $\Delta = 0$, then we have $\sigma = 0$, in which case the ILC update~\eqref{eq:ilc_normal} achieves sublinear regret. 
This is explained by the fact that the ILC fixed point $\bar{x}$ and $x^*$ coincide in this case as the ILC update uses the true gradient information.
Therefore, by improving on the fixed point by reducing $d$, it is possible to achieve sublinear regret. The following proposition illustrates another special case featuring sublinear regret.
\begin{proposition}\label{prop:ilc_normal_sublinear}
	Let Assumptions \ref{assm:model} and \ref{assm:Lips} be satisfied, and assume further that $\textstyle{\gamma <1}$.
	Consider the choice of preconditioner $\textstyle{W = M^TQM}$, cost $\textstyle{f(x) = \frac{1}{2} || H x - r||^2_Q}$, assume that $\textstyle{\min_{x\in\U}f(x) = 0}$
	and a constant step-size is chosen as $\textstyle{\alpha=\alpha_0 \in \left(0,\frac{2}{1+w\gamma}\right)}$.
	Then, the static/dynamic regret for the controller update \eqref{eq:ilc_normal} is bounded by
	\begin{align}
		J_{\mbox{ILC}}(T) \leq  \frac{\bar{L} \delta_{x_1} }{1-\phi},
	\end{align}
	where $\delta_{x_1}:=||x_1-x^*||_W$, $\textstyle{x^* = \arg\min_{x\in\U}f(x)}$ and $\phi = ||I - \alpha_0 W^{-1} (M^TQH)||_W$.
	Hence,
	\begin{align}
		\lim\limits_{T \to \infty}	\frac{J_{\mbox{ILC}}(T)}{T}  = 0
	\end{align}
\vspace{0.01cm}
\end{proposition}
The proof of Proposition~\ref{prop:ilc_normal_sublinear} follows from Proposition~\ref{prop:ilc_normal} by recognizing that we have here $\textstyle{||W^{-1} (M\Delta)^TQ(H x^* - r) ||_W = 0}$ by assumption, since $\textstyle{\min_{x\in\U}f(x) = 0}$ implies that $\textstyle{Hx^* - r = 0}$. Therefore, by having a small enough disturbance set, i.e., $\textstyle{\gamma<1}$, and assuming that the optimal input is feasible for the true dynamics ($\textstyle{Hx^* - r = 0}$), the ILC update~\eqref{eq:ilc_normal} has sublinear regret. 
We note that, while these assumptions are much stronger than those employed in this paper so far, they are provided as edge cases that may apply in certain scenarios. 

Following Corollary~\ref{cor:dyn_regret_asymp_2}, it is easy to see that the regret of the ILC update \eqref{eq:ilc_normal} becomes sublinear if model learning takes place concurrently with the controller iterations and Assumptions~\ref{assm:model_2} and \ref{Ass_sigma_2} are satisfied for the iteration invariant problem.
Note that the sublinear regret condition in this case is achieved without assuming any convergence rate for $\gamma_k$, as long as we have $\textstyle{||\bar{x}-x^*||\to 0}$ asymptotically. 
For example, certain model-free ILC applications use past input output measurements to improve model approximation continuously~\cite{rezaeizadeh2017iterative}.
Further analysis of similar approaches for model learning and adaptive POGD-ILC development is subject for future work.

\section{Numerical Demonstration}\label{sec:demo}
For the numerical demonstration we turn to process control for a Selective Laser Melting (SLM) additive manufacturing process.
In SLM, fine metal powder is deposited, melted with the help of a high power laser, and left to solidify in layers, to build a three dimensional object in a layer-by-layer fashion. 
The melt pool dynamics at the point where the laser interacts with the material is of crucial importance for the mechanical properties of the finished part.
Additionally, due to the complex physics and multiple sources of disturbances involved in the process, modeling and controlling the melt pool effectively is an important research challenge.
We use the high-fidelity numerical simulations of an SLM process presented in~\cite{afrasiabi2022sph} to model the melt pool length output as a function of the power input to the system. 
From the high-fidelity simulation data of~\cite{afrasiabi2022sph}, we extract a 5 dimensional discrete-time linear time invariant model of the form 
\begin{align}\label{eq:sim_model}
	\arraycolsep=1.4pt
\mathcal{H} :\left\{ \begin{array}{rl}  \xi(t+1)&= A \xi(t) + B v(t), \\
 y(t)&= C \xi(t),
\end{array}\right.
\end{align}
where $v(t)$ is the instantaneous power input to the system and $y(t)$ is the melt pool length, so that the model is single input single output (SISO).
The constraint set on the input power is defined by the minimum power requirement to initiate melting, and an upper limit based on the constraints of the actuating laser, given by $\textstyle{\mathcal{V} = [75,165]}$, in \textit{Watt}, so that we constraint our input to $\textstyle{v(t)\in\mathcal{V}}$.
Using the model~\eqref{eq:sim_model}, we construct the lifted input-output model of the system for an iteration duration of $100$ time steps, representing a single layer of the SLM process, so that $\textstyle{M\in\reals^{100\times 100}}$.
The input constraint set $\U$ for the POGD-ILC algorithm is then constructed using $\mathcal{V}$.
Following our model assumption, we compute the true input-output dynamics in each iteration as $\textstyle{H_k = M + M\Delta}$, where $\Delta$ has a diagonal structure, and is sampled from the set $\bm{\Delta}(W,\gamma_0)$ for each $k$. 
We present results (i) for the dynamic regret of Non-Adaptive POGD-ILC with diminishing step-size $\textstyle{\alpha_k = \alpha_0 k^{-c}}$, chosen according to Corollary~\ref{cor:dyn_regret_asymp}, and (ii) for the case with model learning Adaptive POGD-ILC, a constant step size chosen according to Corollary~\ref{cor:dyn_regret_asymp_2}.
For the adaptive POGD-ILC, we emulate the adaptation by a diminishing uncertainty set and $\textstyle{\gamma_k  = \gamma_0k^{-1/2}}$, where $\gamma_0$ is the initial uncertainty, also used in the non-adaptive case.

\begin{figure}
    \centering
    \includegraphics[width=0.95\columnwidth]{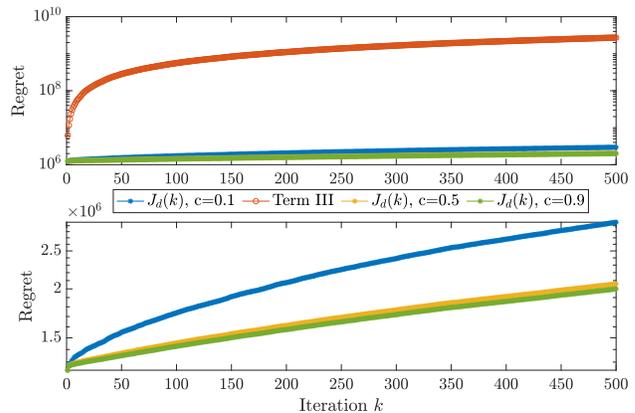}
    \caption{\textit{Top}: Dynamic regret with different step-size rules and Term III from Theorem~\ref{thm:main_dynamic2}.
    \textit{Bottom}: Zoomed plot of dynamic regret curves on a semi-logarithmic scale.}
    \label{fig:dyn_regret_comparison}
\end{figure}

The dynamic regret of the POGD-ILC controller under three step-size choices, differing for the rate of decay, is shown in Fig.~\ref{fig:dyn_regret_comparison}. 
The top plot shows in addition the complexity term (Term III from Theorem~\ref{thm:main_dynamic2}), providing part of the upper bound as predicted analytically. A close-up of the regret progression for the three step sizes is shown on the bottom plot of Fig.~\ref{fig:dyn_regret_comparison}. We see that the regret progression increases with diminishing $c$, which can be explained by the effect of the step size in Term II of Theorem~\ref{thm:main_dynamic2}, suggesting that larger step sizes result in increased regret upper bounds.

\begin{figure}
    \centering
    \includegraphics[width=0.8\columnwidth]{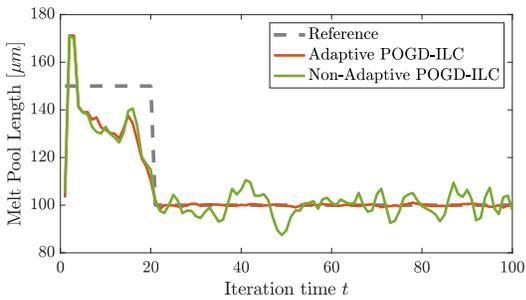}
    \caption{Comparison of the output trajectories at the last iteration $\textstyle{k=500}$, for the adaptive and non-adaptive cases.}
    \label{fig:final_output_comparison}
\end{figure}

The melt pool length output of the two scenarios are illustrated in Fig.~\ref{fig:final_output_comparison}. The tracking performance of the adaptive POGD-ILC is much better due to the model learning and adaptation, while the non-adaptive POGD-ILC still tracks the reference signal, albeit with higher error.

\section{Conclusion}\label{sec:conc}
This work analyzes the regret of online learning iterative learning controllers with model mismatch between the true process and the controller model. We propose a projected online gradient descent controller inspired by online convex optimization methods and analyze the regret performance of the proposed controller under various assumptions and conditions. The results are further extended to the cases related to some of the common ILC schemes from the literature with iteration invariant input-output dynamics. 
Simulation of the dynamic regret performance for the proposed controllers is investigated and numerical evidence is reconciled with the theoretic results. 

While an emulated model learning scheme was adopted in the simulation results to showcase the importance of adaptation, developing and implementing effective methods for model learning is an important research question and subject for future work. 
Moreover, considering additive disturbance models to capture measurement and process noise, and incorporating additional state constraints are important extensions to enable practical uses of the proposed work.

\section*{Acknowledgement}
The authors would like to thank Mamzi Afrasiabi for providing the simulation data used in model development for the numerical example, and the regret analysis reading group at the Automatic Control Lab at ETH Zurich, for useful discussions in conceptualizing this work.

\bibliographystyle{IEEEtran}

\bibliography{references}

\end{document}

%% file: preamble.tex
\usepackage{epsfig} 
\usepackage{amsmath} 
\usepackage{amssymb}  
\usepackage{mathtools,bbm}
\usepackage{xfrac}
\usepackage{verbatim}
\usepackage{bm}
\usepackage{algorithm}

\usepackage{tabularx}
\usepackage{multirow}
\usepackage{adjustbox}
\usepackage{multirow}
\usepackage{float}
\usepackage{enumitem}
\usepackage{amsfonts}
\usepackage{color}
\usepackage[dvipsnames]{xcolor}

\newcommand{\U}{\mathcal{X}}

\usepackage{tikz}
\usetikzlibrary{shapes.geometric, arrows}

\tikzstyle{startstop} = [rectangle, rounded corners, minimum width=1cm, minimum height=0.75cm,text centered, draw=black, fill=green!5]
\tikzstyle{io} = [trapezium, trapezium left angle=70, trapezium right angle=110, minimum width=3cm, minimum height=1cm, text centered, draw=black, fill=orange!15]
\tikzstyle{process} = [rectangle, minimum width=2cm, minimum height=0.75cm, text centered, text width=2.5cm, draw=black, fill=orange!15]
\tikzstyle{process2} = [rectangle, minimum width=1cm, minimum height=1cm, text centered, draw=black, fill=red!20]
\tikzstyle{processLong} = [rectangle, minimum height=0.75cm, text centered,  draw=black, fill=orange!15]
\tikzstyle{decision} = [diamond, minimum width=2cm, minimum height=0.75cm, text centered, draw=black, fill=blue!15]
\tikzstyle{arrow} = [thick,->,>=stealth]
\tikzstyle{sum} = [draw, circle, minimum size=.25cm]

\newcommand{\BEAS}{\begin{eqnarray*}}
\newcommand{\EEAS}{\end{eqnarray*}}
\newcommand{\BEA}{\begin{eqnarray}}
\newcommand{\EEA}{\end{eqnarray}}
\newcommand{\BEQ}{\begin{equation}}
\newcommand{\EEQ}{\end{equation}}
\newcommand{\BIT}{\begin{itemize}}
\newcommand{\EIT}{\end{itemize}}

\newcommand{\reals}{{\mbox{$\mathbb{R}$}}}

\newcommand{\cond}{\mathop{\textrm{cond}}}

\newcommand{\argmin}{\mathop{\rm argmin}}

\newtheorem{theorem}{Theorem}
\newtheorem{corollary}[theorem]{Corollary}
\newtheorem{lemma}[theorem]{Lemma}
\newtheorem{proposition}[theorem]{Proposition}
\newtheorem{assumption}{\textit{Assumption}}


%% file: main.bbl
\begin{thebibliography}{10}
\providecommand{\url}[1]{#1}
\csname url@rmstyle\endcsname
\providecommand{\newblock}{\relax}
\providecommand{\bibinfo}[2]{#2}
\providecommand\BIBentrySTDinterwordspacing{\spaceskip=0pt\relax}
\providecommand\BIBentryALTinterwordstretchfactor{4}
\providecommand\BIBentryALTinterwordspacing{\spaceskip=\fontdimen2\font plus
\BIBentryALTinterwordstretchfactor\fontdimen3\font minus
  \fontdimen4\font\relax}
\providecommand\BIBforeignlanguage[2]{{%
\expandafter\ifx\csname l@#1\endcsname\relax
\typeout{** WARNING: IEEEtran.bst: No hyphenation pattern has been}%
\typeout{** loaded for the language `#1'. Using the pattern for}%
\typeout{** the default language instead.}%
\else
\language=\csname l@#1\endcsname
\fi
#2}}

\bibitem{agarwal2021regret}
N.~Agarwal, E.~Hazan, A.~Majumdar, and K.~Singh, ``A regret minimization
  approach to iterative learning control,'' in \emph{International Conference
  on Machine Learning}.\hskip 1em plus 0.5em minus 0.4em\relax PMLR, 2021, pp.
  100--109.

\bibitem{dixit2019online}
R.~Dixit, A.~S. Bedi, R.~Tripathi, and K.~Rajawat, ``Online learning with
  inexact proximal online gradient descent algorithms,'' \emph{IEEE
  Transactions on Signal Processing}, vol.~67, no.~5, pp. 1338--1352, 2019.

\bibitem{goel2020regret}
G.~Goel and B.~Hassibi, ``Regret-optimal control in dynamic environments,''
  \emph{arXiv preprint arXiv:2010.10473}, 2020.

\bibitem{muthirayan2021online}
D.~Muthirayan, J.~Yuan, D.~Kalathil, and P.~P. Khargonekar, ``Online learning
  for receding horizon control with provable regret guarantees,'' \emph{arXiv
  preprint arXiv:2111.15041}, 2021.

\bibitem{Raginsky_OCO_Adapt}
M.~Raginsky, A.~Rakhlin, and S.~Yüksel, ``Online convex programming and
  regularization in adaptive control,'' in \emph{49th IEEE Conference on
  Decision and Control (CDC)}, 2010.

\bibitem{Hazan_nonStochastic}
E.~Hazan, S.~Kakade, and K.~Singh, ``The nonstochastic control problem,'' in
  \emph{Proceedings of the 31st International Conference on Algorithmic
  Learning Theory}, vol. 117, 2020, pp. 408--421.

\bibitem{hazan2016introduction}
E.~Hazan \emph{et~al.}, ``Introduction to online convex optimization,''
  \emph{Foundations and Trends{\textregistered} in Optimization}, vol.~2, no.
  3-4, pp. 157--325, 2016.

\bibitem{hazan2007logarithmic}
E.~Hazan, A.~Agarwal, and S.~Kale, ``Logarithmic regret algorithms for online
  convex optimization,'' \emph{Machine Learning}, vol.~69, no.~2, pp. 169--192,
  2007.

\bibitem{barton2010norm}
K.~L. Barton and A.~G. Alleyne, ``A norm optimal approach to time-varying {ILC}
  with application to a multi-axis robotic testbed,'' \emph{IEEE Transactions
  on Control Systems Technology}, vol.~19, no.~1, pp. 166--180, 2010.

\bibitem{amann1996iterative}
N.~Amann, D.~H. Owens, and E.~Rogers, ``Iterative learning control for
  discrete-time systems with exponential rate of convergence,'' \emph{IEE
  Proceedings-Control Theory and Applications}, vol. 143, no.~2, pp. 217--224,
  1996.

\bibitem{mishra2010optimization}
S.~Mishra, U.~Topcu, and M.~Tomizuka, ``Optimization-based constrained
  iterative learning control,'' \emph{IEEE Transactions on Control Systems
  Technology}, vol.~19, no.~6, pp. 1613--1621, 2010.

\bibitem{liao2022robustness}
D.~Liao-McPherson, E.~C. Balta, A.~Rupenyan, and J.~Lygeros, ``On robustness in
  optimization-based constrained iterative learning control,'' \emph{arXiv
  preprint arXiv:2203.05291}, 2022.

\bibitem{altin2017iterative}
B.~Alt{\i}n, J.~Willems, T.~Oomen, and K.~Barton, ``Iterative learning control
  of iteration-varying systems via robust update laws with experimental
  implementation,'' \emph{Control Engineering Practice}, vol.~62, pp. 36--45,
  2017.

\bibitem{balta2020switch}
E.~C. Balta, D.~M. Tilbury, and K.~Barton, ``Switch-based iterative learning
  control for tracking iteration varying references,''
  \emph{IFAC-PapersOnLine}, vol.~53, no.~2, pp. 1493--1498, 2020.

\bibitem{yu2017robust}
M.~Yu and C.~Li, ``Robust adaptive iterative learning control for discrete-time
  nonlinear systems with time-iteration-varying parameters,'' \emph{IEEE
  Transactions on Systems, Man, and Cybernetics: Systems}, vol.~47, no.~7, pp.
  1737--1745, 2017.

\bibitem{van2016optimality}
J.~Van~Zundert, J.~Bolder, and T.~Oomen, ``Optimality and flexibility in
  iterative learning control for varying tasks,'' \emph{Automatica}, vol.~67,
  pp. 295--302, 2016.

\bibitem{hoelzle2015spatial}
D.~J. Hoelzle and K.~L. Barton, ``On spatial iterative learning control via
  {2-D} convolution: Stability analysis and computational efficiency,''
  \emph{IEEE Transactions on Control Systems Technology}, vol.~24, no.~4, pp.
  1504--1512, 2015.

\bibitem{van2009iterative}
J.~Van~de Wijdeven, T.~Donkers, and O.~Bosgra, ``Iterative learning control for
  uncertain systems: Robust monotonic convergence analysis,''
  \emph{Automatica}, vol.~45, no.~10, pp. 2383--2391, 2009.

\bibitem{son2015robust}
T.~D. Son, G.~Pipeleers, and J.~Swevers, ``Robust monotonic convergent
  iterative learning control,'' \emph{IEEE Transactions on Automatic Control},
  vol.~61, no.~4, pp. 1063--1068, 2015.

\bibitem{Jadbabaie_ICAIS_15}
A.~Jadbabaie, A.~Rakhlin, S.~Shahrampour, and K.~Sridharan, ``{Online
  Optimization : Competing with Dynamic Comparators},'' in \emph{Proceedings of
  the Eighteenth International Conference on Artificial Intelligence and
  Statistics}, ser. Proceedings of Machine Learning Research, vol.~38, 2015,
  pp. 398--406.

\bibitem{martin2022safe}
A.~Martin, L.~Furieri, F.~D{\"o}rfler, J.~Lygeros, and G.~F. Trecate, ``Safe
  control with minimal regret,'' \emph{arXiv preprint arXiv:2203.00358}, 2022.

\bibitem{Hall_Willett_ICML_2013}
E.~Hall and R.~Willett, ``Dynamical models and tracking regret in online convex
  programming,'' in \emph{Proceedings of the 30th International Conference on
  Machine Learning}, vol.~28, no.~1, 2013, pp. 579--587.

\bibitem{astrom2008adaptive}
K.~{\AA}str{\"o}m and B.~Wittenmark, \emph{Adaptive Control}, ser. Dover Books
  on Electrical Engineering.\hskip 1em plus 0.5em minus 0.4em\relax Dover
  Publications, 2008.

\bibitem{rezaeizadeh2017iterative}
A.~Rezaeizadeh and R.~S. Smith, ``Iterative learning control for the radio
  frequency subsystems of a free-electron laser,'' \emph{IEEE Transactions on
  Control Systems Technology}, vol.~26, no.~5, pp. 1567--1577, 2017.

\bibitem{afrasiabi2022sph}
M.~Afrasiabi, C.~Luthi, M.~Bambach, and K.~Wegener, ``Smoothed particle
  hydrodynamics modeling of the multi-layer laser powder bed fusion process,''
  \emph{Procedia CIRP}, 2022.

\end{thebibliography}
